\documentclass{article}
\usepackage{array}
\usepackage{tabularx}
\usepackage{booktabs}
\usepackage{ragged2e}
\usepackage{macros}
\title{Different Forms of Imbalance in Strongly Playable Discrete
 Games II: Multi-Player RPS Games}
\author{Itai Maimon}
\affil{Department of Mathematics \\
University of California San Diego \\
La Jolla, CA 92093-0112 (USA) \\
\{\tt imaimon\}@ucsd.edu}
\date{\today}

\begin{document}

\maketitle
\begin{abstract}

Classic Rock-Paper-Scissors, $RPS$, has seen many variants and generalizations in the past several years. In the previous paper, we defined playability and balance for games. We used these definitions to show that different forms of imbalance agree on the most balanced and least balanced form of playable two-player $n$-object $RPS$ games, referred to as $(2,n)$-$RPS$. We reintroduce these definitions here and show that, given a conjecture, the generalization of this game for $m<50$ players is a strongly playable $RPS$ game. We also show that this game maximizes these forms of imbalance in the limit as the number of players goes to infinity.

\noindent

$\,$

\noindent
\end{abstract}

\tableofcontents

\newpage

\section{Introduction: Constructions of Imbalance, Playability and Blow-ups}
Classic Rock-Paper-Scissors, $RPS$, has seen many variants and generalizations in the past several years \cite{SpiroSuryaZeng2022_SemiRestrictedRPS}. In the previous paper, we defined playability and balance for games and used these definitions to show that different forms of imbalance agree on the most balanced and least balanced form of playable two-player $n$-object $RPS$ games, referred to as $(2,n)$-$RPS$
\cite{PreviousPaper}. We reintroduce these definitions here and show that a generalization of this game for $m<50$ players is a strongly playable $(m,n)$-$RPS$. We also show that this game maximizes these forms of imbalance in the limit as the number of players goes to infinity.

Intuitively, a maximally imbalanced $RPS$ game should contain an object that beats all other objects. In this case, the only undominated strategy would be to play said object. On the other hand, broadly, we say an $RPS$ game is playable if there is a Nash equilibrium such that each object has a positive probability of being played. This opposition of definitions leads us to the question of what is the maximally imbalanced strongly playable $(m,n)$-$RPS$ game. In the two-player case, we identified this game and used its existence to justify a connection between imbalance statistics for games and inequality statistics for populations, in that while measures can broadly disagree \cite{Cowell2000_MeasurementOfInequality}, they should agree on the maximally balanced/equitable and minimally balanced/least equitable examples, respectively.

\subsubsection*{Definitions of Multiplayer $RPS$ Games}
An $(m,n)$-$RPS$ game is a totally symmetric discrete zero-sum win-lose game. Here we say that a game is win-lose if each player either wins, receiving the same payoff as all other winners, or loses, receiving the same payoff as all other losers. In each step of an $(m,n)$-$RPS$ game, each player chooses one of the $n$ objects. For any multi-set of choices, the rules of the step determine a unique winning object. For instance, a $(3,3)$-$RPS$ may contain objects $R,P,S$ and players $l_1,l_2,l_3$. In this game, each multi-set of three objects $\{\!\!\{R, R, R\}\!\!\}, \{\!\!\{R, R, P\}\!\!\},\dots $ determines a unique object that wins that step. For instance, if for multi-set $\{\!\!\{R, R, P\}\!\!\}$, object $R$ wins, then in an instance of this game where $l_1$ chose $P$ and $l_2,l_3$ chose $R$ then $l_1$ would lose and both $l_2$ and $l_3$ would win. 

We can then define a step of this game by a map, $\phi$, that chooses a winner given a multi-set. More precisely, $\phi$ is a map from the collection of multi-sets of $m$ objects to the set of all objects, such that for a multi-set $c$, $\phi(c)\in c$. When two or more players choose the same winning object, we say they tied for the win. In the above example, $l_2,l_3$ would tie for the win, and $l_1$ would lose. As $RPS$ games usually are replayed in the event of a tie, if we have $m'$ tied winners, those players would take another step by playing with $n$ objects and $m'$ players until they determine a unique winner.

In an instance of the game, only a single player wins, and $(m-1)$ players lose. For the game to remain sum-zero, without loss of generality, wins provide a pay-off of $(m-1)$ and losses a payoff of $-1$. In the case of an $m'$-way tie after a single step, the expected payoff of the second step is $\frac{m-m'}{m'}$. As this is a symmetric game without collusion, a priori, of a choice of known strategies, the outcome of all later steps is equivalent to selecting a winner uniformly at random. We can then define the effective payoff tensor of the first step of an $(m,n)$-$RPS$ without fully considering the actions in later steps \cite{tensornplayer}. This alternative payoff leads to the following definition of a $(m,n)$-$RPS$ that we will use for the remainder of the discussion.

\begin{defi}[ 
Alternative form of $(m,n)$-$RPS$]
We define an alternative form of $(m,n)$-$RPS$, $G$, as a symmetric win-lose sum-zero game with $m$ players and $n$ pure strategies, played without collusion between any players. Where for an $m'$ way tie all winners receive a payoff of $\frac{m-m'}{m'}$ and all losers receive a payoff of $-1$.

\end{defi}

\subsection{Playablity and Imbalance}

Given the above construction of $(m,n)$-$RPS$, we can now define what it means for an $RPS$ game to be playable as described in the prequel.

\begin{defi}[$k$-Playable]
For $k>0$, an $RPS$ game is $k$-playable if a Nash equilibrium, $N$, exists, such that each object has a positive probability of being played by $k$ players. In other words, in $N$, for each object, $o$, there are at least $k$ players, $p^1_o\dots p^k_o$, whose mixed strategy has a positive probability of playing $o$. In this case, we say that $o^k$ sees play by $p^1_o\dots p^k_o$ in Nash equilibrium $N$. A game is $k$-strongly playable if each Nash equilibrium has this property. On the other hand, a game is $k$-weakly playable if for each object $o$ there is a Nash equilibrium, $N_o$, in which $o^k$ sees play. We say a game is playable if it is $1$-playable, and similarly define a game as strongly/weakly playable if it is $1$-strongly/$1$-weakly playable.
\end{defi}

 The reasoning behind these definitions is to restrict ourselves to games in which each pure strategy is valid for some player to consider under a given meta-game. The three different forms of playability then define three different ways an $RPS$ game can have every strategy \textit{see play} in a competitive environment. For weakly playable games, each strategy has a meta-game in which it is played with positive probability. In contrast, for playable games, there is a meta-game such that each strategy is played with a positive probability. Finally, in strongly playable games, each meta-game must have each strategy played with positive probability. In the prequel, we further generalized this definition for possibly asymmetric and non-discrete games.

The imbalanced statistics that we defined in the prequel are based on the idea that in a balanced symmetric game, each strategy should be as good as any other. Thus, each player should be able to play each object with a uniform distribution. Our first type of imbalance statistic, referred to as uniform strategy imbalances, is based on the distribution of the expected value for the remaining player if all other players play with a uniform strategy. We call these the uniform expected payoffs of each strategy. In a balanced game, in this scenario, we would expect each object to have the same payoff. Therefore, if we consider these expected values as a distribution giving equal weight to each object's uniform expected payoff, then the balanced game should have a Dirac-delta distribution of expected payoffs. Thus, an imbalanced measure can be defined as a measure of distance away from the Dirac-delta distribution. 

This intuition leads us to the following three definitions:

\begin{defi}[Uniform Imbalance of Variance]
For two symmetric games, $G_1$ and $G_2$, $G_1$ is more uniformly imbalanced in variance, referred to as $UI_v$, than $G_2$ if the variance of the uniform expected payoffs in $G_1$ is higher than the variance of the uniform expected payoffs in $G_2$. 
\end{defi}

\begin{defi}[Uniform Imbalance of Entropy]
 For two symmetric games, $G_1$ and $G_2$, $G_1$ is more uniformly imbalanced in entropy, referred to as $UI_e$, than $G_2$ if the entropy of the uniform expected payoffs as a random variable in $G_1$ is higher than the entropy of the uniform expected payoffs in $G_2$. Where the entropy, $H$, of a probability distribution, $P$, over a measurable set $D$, is defined as \cite{Shannon1948MathematicalTheoryCommunication}:
 \begin{equation}
     H(P)=-\int_D P(x)\ln(P(x))dx
 \end{equation}
\end{defi}

\begin{defi}[\alpha-Thiel Entropy Imbalance ]
For a symmetric game, $G_1$ with bounded payoffs, define $\tilde{G}$ as a game in which we scale the payoffs of each outcome in $G$ by a positive constant $c_1$ and then add a constant $c_2$ to each payoff, so that the uniform expected payoffs has mean $1$ and minimum value $0<\alpha<1$. Note that this transformation does not change the Nash Equilibria from $G_1$ to $\tilde{G}_1$. 
For two games with bounded payoffs, $G_1$, $G_2$, $G_1$ is more uniformly imbalanced in $\alpha$-Theil $T$ entropy, referred to as $UI_{t_\alpha}$, than $G_2$ if the Thiel-T index of the random distribution over uniform expected payoffs for $\tilde{G_1}$ is higher than the Thiel-T index over $\tilde{G}_2$. Where we define the Thiel-T index of a distribution, $P(x)$ over positive values with mean $1$ as:
\begin{equation}
    T(P)=\int_0^\infty P(x) x\ln(x) dx
\end{equation}
Over discrete random variables defined on a set $D$, with probability distribution $P(x)$ as:
\begin{equation}
    T(P)=\sum_{i\in D} P(x_i)x_i\ln(x_i)
\end{equation}
\end{defi}

The second form of imbalance definitions is based on the Nash equilibria of the underlying game. A game that contains a Nash equilibrium where each player plays a uniform distribution over all strategies should be considered quite balanced. We therefore define the Nash probability imbalances as those based on the worst-case (most-balanced) Nash equilibrium's deviation from uniformity. We similarly define the symmetric Nash probability imbalances in the same way, restricted to only symmetric Nash equilibria. This leads to the following two definitions:

\begin{defi}[Minimal Nash Equilibria Entropy Imbalance]
     Given two $m$-player games, $G_1$, and $G_2$, $G_1$ is more Nash equilibria entropy imbalanced, referred to as $N_e$, than $G_2$ if in the Nash equilibrium, which maximizes the sum of the entropies for all players of $G_1$, and $G_2$, $n_1$, $n_2$ respectively, $H(n_1)<H(n_2)$. Where $H(n_i)$ is the sum of the entropy of each distribution of each player in the set of distributions $n_i$. For instance, if $P_j(x)$ is the distribution over pure strategies for player $p_j$ in $n_i$, then given $m$ players, $p_1 \dots p_m$, each with their own strategy set $D_j$: 
     
     \begin{equation}
         H(n_i)=-\sum_{j=1}^m \int_{D_j} P_j(x)\ln(P_j(x))dx
     \end{equation} 
\end{defi}

When playing an imbalanced game, we would expect each player to choose the most imbalanced object, and consequently, tie often. Conversely, we would generally expect that in a balanced game, multiple players would choose the same option much less frequently. We can then state that the number of ties in our symmetric equilibria of a balanced game should be minimized, and the number of ties in such an equilibrium should be maximized in imbalanced games. This leads to the following definition:
\begin{defi}[Maximal Nash Equilibria Ties Imbalance]
     Given two $m$-player symmetric discrete games, $G_1$, and $G_2$, $G_1$ is more Nash maximal ties imbalanced, $N_t$, than $G_2$ if in the symmetric Nash equilibrium, which minimizes ties for $G_1$, and $G_2$, $n_1$, $n_2$ respectively, $T(n_1)>T(n_2)$, where $T(n_i)$ is the expected number of ties when players play the set of strategies $n_i$ over $d$ pure strategies. For instance, if $v_{o}$ is the probability of a player playing $o$ in the Nash equilibrium $n_i$ then, 
     \begin{equation}
         T(n_i)=\sum_{o=1}^d {v_o}^m
     \end{equation}   
\end{defi}     

As many of these imbalance statistics are Schur-convex in their variables \cite{Majorization2e}, we can define classes of imbalance statistics, which we call Schur-uniform or Schur-distributional imbalance statistics, respectively. More precisely, we can define:

\begin{defi}[Multiplayer Schur-Uniform Class for $RPS$ games]
    A uniform strategy imbalance, $\tau$, is in the multiplayer Schur-uniform class, $S_u$, if:

    \begin{enumerate}
        \item $\tau$ is definable over an $(m,n)$-$RPS$ game for any $n,m\in \N$.
        \item  When $\tau$ is applied to the $(m,n)$-$RPS$ games, $G_1$, and $G_2$ with sequences of uniform expected payoffs $\mathbf{P}_{G_1,in},\mathbf{P}_{G_2,in}$ respectively, if $\mathbf{P}_{G_1,in}$ majorizes $\mathbf{P}_{G_2,in}$ then $G_1$ is greater than or equally as imbalanced as $G_2$.
    \end{enumerate}
    If, when $\mathbf{P}_{G_1,in}$ strictly majorizes $\mathbf{P}_{G_2,in}$, $G_1$ is strictly more imbalanced than $G_2$ we say that said imbalance definition is in the strict $S_u$ class.
\end{defi}

Similarly, we have:

\begin{defi}[Multiplayer Schur-Distributional Class for $RPS$ games]
    A Nash probability imbalance, $\tau$, is in the Schur-distributional class, $S_N$ or symmetric Schur-distributional Class $S_{N_s}$, if:

    \begin{enumerate}
        \item $\tau$ is definable over a $(m,n)$-$RPS$ game for any $(m,n)\in \N$
        \item  When $\tau$ is applied to  $(m,n)$-$RPS$ games, $G_1$, and $G_2$ with worst-case Nash equilibrium/symmetric Nash equilbrium probabilities for player $j$ playing object $o$ of $n^j_{G_1,o}$ and $n^j_{G_2,o}$ respectively, if the sequence of all such player-object probabilities for $G_1$ majorizes that sequence for $G_2$ then $G_1$ is greater than or equally imbalanced as $G_2$.
    \end{enumerate}
    If, when a sequence of probabilities for $G_1$ strictly majorizes $G_2$'s sequence of probabilities implies $G_1$ is more imbalanced than $G_2$, we say that said imbalance definition is in the strict $S_N$/$S_{N_s}$ class.
\end{defi} 

All but uniform imbalance in entropy, $UI_e$, are in these Schur classes of imbalance definitions, as seen by the fact that the functions defining these statistics are symmetric and convex, respectively.

Therefore, we can show that a game is very imbalanced if the uniform expected payoffs in a playable $RPS$ game are close to the maximally majorizing set of expected payoffs for such games. Similarly, it maximizes the Nash distributional/symmetric Nash distributional imbalance if the set of worst-case Nash equilibria player probabilities, or symmetric Nash equilibria player probabilities, is close to the maximally majorizing set of such probabilities.

 Generally, one can consider a game, $G$, more imbalanced than $H$ if the \textit{good} objects in $G$ win in more scenarios than the \textit{good} objects of $H$, and the \textit{bad} objects in $G$ lose in more scenarios than the \textit{bad} objects in $H$. It is easy to make arbitrarily imbalanced games; however, in playable games, all objects can \textit{be played}, and so, there must be a Nash equilibrium where even the \textit{worst} object must have a probability of being played.

\subsection{Constructions of Multiplayer Blow-Ups}

In \cite{PreviousPaper}, we proved that there are no playable two-player $RPS$ games with an even number of objects. We also found a playable $(2,2n+1)$-$RPS$ that is least balanced according to all given imbalance definitions. This game is defined on a set of objects $r_0\dots r_{n-1},p_0,\dots p_{n-1},s=r_n=p_n$, such that, $r_i$ beats all elements $r_j,p_k$ for $i>j,i>k$, $p_i$ loses to all elements $r_j,p_k$ for $i>j,i>k$, and lastly $p_i$ beats $r_i$. The structure of this game can be simplified by using the following definition of a blow-up:

\begin{defi}[Symmetric Blow-up]
The blow-up of two symmetric discrete two-player games, $G, H$, over $l\in G$, denoted as $G\#_{l} H$, is a game with the pure strategies in $(G\setminus \{l\}) \cup H$. If we let $(G\setminus l)$ be the payoff matrix of $G$ without the row and column corresponding to $l$, and $G_l$ be the $(|G|-1)\times |H|$ matrix with each column being a copy of the column corresponding to $l$ in the payoff matrix of $G$, then $G\#_{l} H$ has payoff matrix:
        \begin{equation}
            \begin{bmatrix}
                (G\setminus l) & G_l \\
                (G_l)^\top& H 
            \end{bmatrix}
        \end{equation}
        
        This payoff matrix implies that: 
        \begin{itemize}
            \item If both players play strategies in $H$, the payoff matrix is the payoff matrix in $H$.
            \item  If both players play in $G$, with neither player playing $l$, then the payoff matrix is the payoff matrix in $G$.
            \item If player $p$ plays in $k\in H$, and $q$ plays $i\in G\setminus l$ the payoff is the same as those in $G$ if $p$ played $l$ and $q$ played $i$, and vice versa. 
        \end{itemize} 
\end{defi}

We can see that blow-ups are related to both modular products of graphs and lexicographical products of games \cite{Salha2022MaxDecomposability}\cite{GoldbergMoon1970CompositionTournaments}. Using blow-ups we can construct the imbalanced $(2,2n+1)$-$RPS$ as: $3$-$RPS\#_{S}(3$-$RPS\#_{S}\dots)$ $n$ times. This construction reveals that the least balanced playable $(2,2n+2m+1)$-$RPS$ can be constructed by blowing up the least balanced  playable $(2,2n+1)$-$RPS$, $G$, with the least balanced but playable $(2,2m+1)$-$RPS$, $H$, at the most balanced object, $s\in G$.

Unlike in the two-player case, there is no canonical blow-up in a game with multiple players, which comes about because the distribution of payoffs is less well-defined. For instance, consider a four player form of  $G\#_l H$, if two players are playing in the sub-game, $H$, and two others are playing objects in $G\setminus l$, it is unclear if we should use the results as if the two players played $l$ or as if the two players played in $H$, or some mix of the two. It is also clear that blow-ups for $m$-player games will only be possible if the sub-game can be played with any subset of players, so that the outcomes in the subgame could be well-defined. We provide the following definition of a win/lose game and an $m$-fully extended game in order to make the construction:

\begin{defi}[Win/Lose Games]
    A win/lose game is one in which a player, $q$, either wins or loses. The outcomes for all the winners are the same, as are the payoffs for all the losers. An all-way tie is, by convention, a win for all players. 
\end{defi} 

\begin{defi}[Fully-Extended Games]
    An $m$-fully-extended game, $G$, is a game that is playable with any subset of $m$ given players. Precisely, a game $G$, with players $p_1 \dots p_m$, and corresponding strategy sets, $S_1,\dots, S_m$, is a map from each finite subset of players, $U=\{p_{u_1}\dots p_{u_k}\}$, to a set of payoff tensors for each player between the other players in the subset $\{\ T^{p_{u_i}}_{p_{u_1},\dots p_{u_{i-1}}, p_{u_i},\dots p_{u_k}}\dots \} $. These payoff tensors describe the payoffs for each player in $U$ if only those given players played the game $G$. 
    \end{defi}
For completeness, an $m$-fully-extended  game, $G$, is said to contain an $m'$-fully extended game, $H$, if there is a subset $U_H$ of players of $G$ of size $m'$ such that there is a bijection, $\phi$, from the strategy set of the players in $U_H$ and the players in $H$ that respects the payoff tensors between players $U_H \subset G$.

An $m$-fully-extended  win/lose symmetric game, $G$, with strategy set, $S$, can be described by a map, $g$, from the collection of all finite multi-sets of $S$ of size up to $m$, $s_m$ to $P(S)\times \R\times \R$, where $P(S)$ is the power set of $S$ which defines the set of winners, $W$. For example, say:
\begin{equation}
    g(c)=(W_c,P_{W_c},P_{L_c})
\end{equation}
Then any player who played an element of $W_c\subseteq C$ wins a payoff of $P_{W_c}$ and otherwise players receive a payoff of $P_{L_c}$.
For instance, odd-one out is a symmetric win/lose $\infty$-extended game. This game contains two pure strategies, $a$, and $b$, such that if $m'$ players choose $a$, and $m''$ players choose $b$, then if $m'<m''$, the players who chose $a$ win. If the reverse occurs, then the players who chose $b$ win. If the numbers are precisely the same, then it is an all-way tie.

These definitions allow us to provide the following construction of the multi-player blow-up:

\begin{defi}[ Multi-player Blow-Up]
For a sum-zero $m$-player win/lose game, $G$, and sum-zero $m$-fully-extended, win/lose game, $H$ the blow-up at $l_1\dots l_m\in G$ is denoted as $G\#_{(l_1,\dots l_m)}H$. Player $p_i$'s set of pure strategies in $G\#_{l_1,\dots l_m}H$ is then the union of their pure strategies in $G$, excluding $l_i$, and their pure strategies in $H$. For any set, $O$, of pure strategies chosen by players in an instance of the blown-up game, consider the set $O_G$ as the adjusted strategies by taking each player $p_j$'s strategy in $O\cap H$ and replacing it with $l_j$ accordingly. Then those that chose winning strategies in $O_G$ as a collection of strategies for the game $G$, excluding those choosing strategies $l_1\dots l_m$, won this game. Given the subset $L$ of $\{l_1\dots l_m\}$ of strategies which won in $O_G$, define $O_H$ as the corresponding chosen strategies of $L$ as strategies in $H$. Then, those who chose the winning strategies in $O_H$ as a collection of strategies in $H$ are also winners of this game. By convention, the losers each recieve a payoff of $-1$ and each of the $m'$ winners recieve a payoff of $\frac{m-m'}{m'}$ 
\end{defi}

We can simplify this definition substantially if both games and the blow-up are symmetric. The set of pure strategies of this new game is the set of pure strategies in $(G\setminus l) \cup H$. Given an $m$-multi-set of pure strategies, $O=\{o_i\}_{i=1}^m$ as the chosen strategies, consider the associated multi-set in $G$, $O_{G}$, taken by replacing every strategy in $O\cap H$ with $l$, and the associated multi-set of strategies of $H$, $O_H=O\cap H$. If in $O_{G}$ as a collection of strategies for $G$, $l$ loses or is not played, then the winners are the winners in $O_{G}$. Otherwise, the winners are the winners in $O_G$ and the winners in $O_H$ as a collection of strategies for $H$.

It would be interesting to find when the Nash equilibria of the component games extend to a Nash equilibrium of the composite game. For instance, in the game of musical chairs, we can blow up each chair \textit{at the loser}. In this case, instead of both people who played the same chair losing, if they choose the blown-up chair, the two players instead determine the loser through the choice of strategy in the subgame. This is opposite to the blow ups above which can be considered as blow ups \textit{at the winner}. If we were to blow up musical chairs at each chair, it would be reasonable to expect that, depending on the choices of sub-games on each chair, the sub-games may have different, strange Nash equilibria, whose projection to the original game of musical chairs may not be a Nash equilibrium at all.

\section{An Imbalanced $(m,3)$-RPS Construction}
Here, we will construct a strongly playable imbalanced $(m,3)$-$RPS$, and in the following section, we will use it and the blow-up construction to provide an imbalanced $(m,2n+1)$-$RPS$. We will show that the given $(m,2n+1)$-$RPS$ is strongly playable for at least $50$ players, and conjecture that this is also the case for any number of players. We motivate that this $(m,3)$-$RPS$ game is nearly maximally imbalanced by showing that the uniform expected payoffs converge to one that majorizes those of all other such games in the limit as the number of players goes to infinity. We similarly motivate that this game majorizes all possible symmetric Nash probabilities in the limit.

The construction of the imbalanced $(m,3)$-$RPS$ starts similarly to that of the two-player case. Let $R$ be the object that we will construct so as to win most multi-sets. As in the two-player case, there needs to be an object that beats $R$ in many situations. For this purpose, let $P$ be an object such that in all multisets containing only $R$ and $P$, $P$ always wins. To maximize $R$'s imbalance in all other multi-sets containing $R$, $R$ should always win. This structure implies that if a player chooses $R$, and even a single opponent chooses the remaining object, $S$, then object $R$ must win. Since the remaining multi-sets contain only $S$ and $P$, for $S$ to have a reason to be played, $S$ should always win in these multi-sets. More precisely, we have defined this game as:

\begin{defi}[Imbalanced  $(m,3)$-$RPS$]
    The imbalanced $(m,3)$-$RPS$ is defined by valuing any multi-set with at least one $S$ and one $R$ as a win for $R$; any multi-set with only $R$ and $P$ as a win for $P$; and any multi-set with only $P$ and $S$ as a win for $S$.
\end{defi}

We can then show the following lemmas:

\begin{lem}[Strong Playability]
    The imbalanced $(m,3)$-$RPS$ is strongly playable.
\end{lem}

\begin{proof}
In any Nash Equilibrium of this game, $N$, if no player has any probability of playing $S$, then all players must play $P$ with probability $1$, as this strictly dominates $R$. However, in the case where all players play $P$ with probability $1$, playing $S$ is a strict advantage, and so $N$ is an invalid Nash equilibrium. Similarly, suppose, in a Nash equilibrium, $N$, no player has any probability of playing $R$. In that case, all other players must play $S$ with probability $1$, which makes playing $R$ a strict advantage, invalidating $N$ as a Nash equilibrium. Finally, if no players play $P$ in $N$, all players would play $R$ with probability $1$, again making playing $P$ a strict advantage, invalidating $N$ as a Nash equilibrium. 
\end{proof}

While the above proof holds, this is not necessarily the most imbalanced strongly playable $(m,3)$-$RPS$. Consider, for instance, an $(m,3)$-$RPS$ with similar rules and very large $m$. The only change in the rules is that in the multi-set with a single $P$ and $m-1$ $R$ objects chosen, $R$ wins. It is not clear that in this case, this game is no longer strongly playable; however, it is more $UI_v$ imbalanced. Therefore, there may be incremental improvements to the maximal imbalance in strongly playable $(m,n)$-$RPS$ games. Instead of showing that this is the least balanced, we show the following lemma:

\begin{lem}[The given $(m,3)$-$RPS$ is Schur-uniformly maximally imbalanced in the limit]
    The given $(m,3)$-$RPS$ has a uniform expected payoff that converges to one that majorizes all other uniform expected payoffs in the limit as $m$ goes to infinity. 
\end{lem}

\begin{proof}
Consider the maximally imbalanced $(m,3)$-$RPS$. This game would contain an object, $R'$, such that in any multi-set containing $R'$, the players choosing $R'$ would win. Similarly, it would contain an object, $P'$, such that in any multi-set containing $P'$ and at least one other object, $P'$ would lose. The final object $S'$ is then defined to lose whenever it is chosen in a multi-set containing $R'$ and win whenever it is in a multi-set containing only $P'$ and $S'$. 

This game is trivially unplayable; however, we can compare its imbalance with that of the game constructed above. First, we can note that in these games, $S'$ acts identically to $S$, so its uniform expected payoff is identical. We can see that the number of multi-sets that $R'$ wins that $R$ would not be the $m-1$ multi-sets containing both $R'$ and $P'$, and these are the same multi-sets that $P'$ would lose that $P$ wins instead.

We can then create an upper bound of the difference of the uniform expected payoffs, $F$, of $R$ and $R'$ as:

\begin{equation}
    F(R')-F(R)< m\left(\frac{2^{m-1}-1}{3^{m-1}}\right)
\end{equation}
Here $\frac{2^{m-1}-1}{3^{m-1}}$ is the probability of losing in $R$ where you would have won in $R'$, and it is multiplied by the largest the difference in payoff from winning and losing can be, $m=m-1-(-1)$.

Similarly, by the fact that these uniform expected payoffs sum to $0$, we have that:
\begin{equation}
    F(P')-F(P)>-m\left(\frac{2^{m-1}-1}{3^{m-1}}\right)
\end{equation}
As $m$ gets large, the difference in these uniform expected payoffs goes to zero at a rate of approximately $(\frac{2}{3})^m$. Therefore, these uniform expected payoffs converge at an exponential rate to those of the maximally imbalanced game that majorizes all other such payoffs.
\end{proof}

We can also show that:

\begin{lem}[The given $(m,3)$-$RPS$ is symmetric Schur-distributionally maximally imbalanced in the limit]
    For large $m$, the given $(m,3)$-$RPS$ has symmetric Nash probabilities that converge to a distribution that majorizes all others as $m$ goes to infinity.
\end{lem}

\begin{proof}
By Nash's result, there must be at least one symmetric Nash equilibrium \cite{nashnperson}. As it is symmetric and strongly playable, each strategy must have exactly zero expected value for all players. To find this Nash equilibrium, let the variable $r,p,s$ describe the probabilities of player $i$ playing strategies $R, P, S$ respectively. We can then solve the corresponding tensor matrix problem to find values of $r, p, s$ \cite{computnashtensornplayer}. We can then see that the expected value for player $i$ playing $P$ is:

\begin{equation}
\begin{split}
    0& =\sum_{i=1}^{m-1} \frac{m-i}{i}p^{i-1}r^{m-i} \binom{m-1}{i-1}-(1-(1-s)^{m-1}) \\
    &=\sum_{i=1}^{m-1} p^{i-1}r^{m-i} \binom{m-1}{i}-(1-(p+r)^{m-1})\\
    &=\frac{r}{p}\sum_{i=1}^{m-1} p^{i}r^{m-i-1} \binom{m-1}{i}-(1-(p+r)^{m-1})\\
    &=\frac{r}{p}((p+r)^{m-1}-r^{m-1})-(1-(p+r)^{m-1})\\
    &=(p+r)^{m-1}\left(\frac{r}{p}+1\right)-\left(1+\frac{r^{m}}{p}\right)\\
    \end{split}
\end{equation}
By multiplying both sides by $P$ then adding $P+R^m$ to both sides, we get:

\begin{equation}(P+R)^m=P+R^m\end{equation}
Using identical arithmetic for the expected value of player $i$ playing $S$, we get:

\begin{equation}(S+P)^m=S+P^m\end{equation}
Finally, we have the equation: 
\begin{equation}S+P+R=1\end{equation}
Solving for $P$,$R$, and $S$ in the limit as $m$ goes to infinity, gets us that exactly one of $P$, $R$,  or $S$ approaches $1$ and the rest approach $0$. Any of these outcomes corresponds to a distribution that majorizes all others in the limit, achieving the most Schur-distributional imbalance possible. 
\end{proof}

Numerically computed values for different choices of $m$ are given in the table below. Note that with 20 players, this results in over $15$ players tied for the win on an average instance of this game.

\renewcommand\tabularxcolumn[1]{>{\RaggedRight\arraybackslash}p{#1}}
\begin{table}[h]
\centering
\caption[Symmetric Nash equilibria for imbalanced game]{Some Example probabilities in Symmetric Nash equilibria for imbalanced $(m,3)$-$RPS$ in different choices of $m$}
\label{example (m,3)-RPS}
\begin{tabularx}{.9\linewidth}{lccc}
\toprule
&\multicolumn{1}{X}{\hspace{16mm}P(R)}
&\multicolumn{1}{X}{\hspace{16mm}P(P)} 
&\multicolumn{1}{X}{\hspace{16mm}P(S)}\\
\midrule
$m=3$ & $.324$ & $.473$ & $.202$\\

$m=5$ & $.288$ & $.622$ & $.090$\\

$m=10$ & $.212 $& $.760$ & $.027$\\

$m=15$ & $.169$ & $.817$& $.013$\\

$m=20$ & $.142$ & $.850$ & $ .008$ \\
\bottomrule
\end{tabularx}
\end{table}

\section{An Important Theorem for the $(m,2k+1)$-$RPS$}

The following theorem will be necessary for our proof of the strong playability of $(m,2n+1)$-$RPS$ for $m<50$.

\begin{thm}\label{hardprob}
    In every Nash equilibrium of $(m,3)$-$RPS$, with $m<50$, at least two players must choose a mixed strategy, with positive probability of playing object $S$.
\end{thm}
In fact, a slightly weaker form of this is necessary to prove strong playability. Considering a Nash equilibrium, $N$, for the above $(m,3)$-$RPS$ such that only one player, $q$, plays $S$. Moreover, all other players receive at least $\epsilon>0$ less expected payoff by playing $S$ than their current strategy. Then in $(m,2k+1)$-$RPS$, where players other than $q$ play their same strategy in the $(m,3)$-$RPS$. Player $q$ plays all objects in the top-most $RPS$ identically, but distributes his probability in the bottom $2k-1$ objects with relative proportions according to the $(2,2k-1)$-$RPS$ strategy, e.g., the ratio between objects in level $i$ to $j$ is $\frac{1}{3^{i-j}}$. However, instead of playing the bottom-most $S$ object, he allocates its probability to the bottom-most $P$-type object.

If another player were to play in the bottom $(2,2k-1)$-$RPS$, they would maximize their expected payoff when $q$ also plays there by playing $S$ in the Nash equilibrium $N$. By doing so, they can receive an expected payoff that is higher than an even split when $q$ also plays here by a factor of $\frac{1}{3^k}$. Thus, for a large enough $k$, this added expected payoff is less than $\epsilon$. Therefore, playing here would still reduce their expected value. Thus, under this assumption, there is a Nash equilibrium where the bottom-most $S$ is not played and the game is not strongly playable. Broadly, we have shown that if a slightly weaker form of Lemma \ref{hardprob} were not the case, then an $(m,2n+1)$-$RPS$ game cannot be strongly playable.

The reason we only prove this for up to $50$ players is that the method used here demonstrates the theorem when there is no solution to a given system of equations and inequalities. We can provide these inequalities for an arbitrary number of players; however, even after simplification, the resulting equations are complicated to manipulate. After fixing the number of players, the equations become polynomials, and computer algebra systems can prove that these equations and inequalities have no solutions. We conjecture that, in general, there is no solution for any number of players, $m\in \N$.

\begin{proof}
 By contradiction, assume only one player $q_s$ plays object $S$. $q_s$ would never play $R$, as it would always result in a loss. If they were to only play $S$, all other players would play $R$, and $q_s$'s expected value of playing $P$ would be $m-1$. Therefore, $q_s$ must play some mixed strategy with probability $P(q_s=S)$ of playing $S$, and $P(q_s=P)=1-P(q_s=S)$ of playing $P$. Therefore, in all Nash equilibria, $q_s$ must have an equal expected payoff for $P$ and $S$. 
 
 By assumption, any other player plays only objects $R$ and $P$ with some probability. If any player chooses only to play the pure strategy $R$, then all other players would only choose the pure strategy $P$, and it would not be a Nash equilibrium by the same argument as above for player $q_s$. This same argument does not work for a player playing solely the pure strategy $P$. Let $t\geq 0$ players play only $P$, and by adjusting variables $k$, other players play in this game. Then after adjusting variables, the the payoff of tying $k'$ ways with $P$ is now $\frac{k-k'}{k'+t}$, and the payoff of tying $k'$ ways with $R$, and $S$ is $\frac{k+t-k'}{k'}$. Each player playing only $P$ has the same expected value, so when solving for their expected value without loss of generality, we will only consider it for some player $q_p$. The remaining players will be labeled $q_{r_1} \dots q_{r_{k-1}}$. 
 To simplify the expected value equations, let the probabilities of player $q_{r_i}$ playing $R$ be $r_i$, and player $q_s$ playing $S$ be $s$. Let the expected payoff of player $q_{r_i}$ playing $P,R,S$, respectively be $E_i(P),E_i(R),E_i(S)$.  Similarly, let the expected payoffs for $q_s$ and $q_p$ be $E_s(P),E_s(R),E_s(S)$, and $E_p(P),E_p(R),E_p(S)$. Before we begin, we will define $P_i(X)$ as the probability of event $X$ happening for player $q_{r_i}$. For the sake of simplicity, in the equations below, $r$-players refers to players $q_{r_1}\dots q_{r_{k-1}}$. We will then show the following lemma that will make the resulting calculation much simpler:
 \begin{lem}
     If there were to be such a Nash equilibrium, then in that equilibrium, all players who play object $R$ play it with the same probability.
 \end{lem}
 
 To demonstrate this lemma, we must first present all given expected values. By direct calculation, we can then show the following equations:

\begin{equation} \label{Eir-1}
E_i(R)=s\sum^{k-1}_{k'=0}\frac{(k+t+1)-(k'+1)}{(k'+1)}P_i(k' \text{ other $r$-players chose } R)-(1-s)\end{equation}

\begin{equation}\label{Eip-1}
E_i(P)=(1-s)\sum^{k-1}_{k'=0}\frac{(k+t+1)-(k'+t+2)}{(k'+1)+t+1}P_i(k' \text{ other $r$-players chose } P)-s\end{equation}

\begin{equation}\label{Eis-1}
\begin{split}
    E_i(S)=
    & \ P_i(\text{All other $r$-players chose } P)\left((1-s)(k+t)+\frac{s(k+t-1)}{2}\right)\\
    &-(1-P_i(\text{All other $r$-players chose } P))
\end{split}
\end{equation}

\begin{equation}\label{Esp-1}
E_s(P)=\sum^{k}_{k'=0} \frac{(k+t+1)-(k'+t+1)}{k'+1+t}P(k' \text{ $r$-players chose } P)\end{equation}

\begin{equation}\label{Ess-1}
E_s(S)= (k+t)P(\text{all $r$-players chose } P)-(1-P(\text{all $r$-players chose } P))
\end{equation}

\begin{equation}\label{Epr-1}
E_p(R)=s\sum^{k}_{k'=0}\frac{(k+t+1)-(k'+1)}{(k'+1)}P( k' \text{ $r$-players chose } R)-(1-s)
\end{equation}

\begin{equation}\label{Epp-1}
E_p(P)=(1-s)\sum^{k}_{k'=0}\frac{(k+t+1)-(k'+t+1)}{(k'+1)+(t-1)+1}P(k' \text{ $r$-players chose } P)-s
\end{equation}

\begin{equation}\label{Eps-1}
\begin{split}
E_p(S)=&\ P(\text{All $r$-players chose } P)\left((1-s)((k+t)+s\left(\frac{k+t-1}{2} \right)\right)\\
&-(1-P(\text{All $r$-players chose } P))
\end{split}
\end{equation}
Thus, the following conjecture is equivalent to the theorem that we are trying to prove:

 \begin{conj} \label{numericallemmea}
     For any $t\in \N, k\in \N$ There is no set of variables $r_1\dots r_{k-1},s$ each contained in the interval $(0,1)$ such that:
    \begin{equation}\label{i-equations}
    E_i(P)=E_i(R)\geq E_i(S)
    \end{equation}
    \begin{equation}\label{sequations}
    E_s(S)=E_s(P)\end{equation}
    \begin{equation}\label{pequations}
    t>0 \implies E_p(P) \geq E_p(R), E_p(P) \geq E_p(S)
    \end{equation}
 \end{conj}
 We can show that this is true for all $t< 50$ and $ k<50$, thus showing that the lemma is true with $50$ or fewer players.

We will start by simplifying the following:
\begin{enumerate}
    \item $P(\text{All $r$-players chose } P)$
    \item $P_i(\text{All other $r$-players chose } P)$
    \item $P(k' \text{ $r$-players chose } R)$
    \item $P_i(k' \text{ other $r$-players chose } R)$
    \item $P( k' \text{ $r$-players chose } P)$
    \item $P_i(k' \text{ other $r$-players chose } P)$
\end{enumerate}

The first two items on this list are straightforward to solve for:

\begin{equation}\label{p-allrplayp}
P(\text{All $r$-players chose } P)=\prod^{k}_{j=1} (1-r_j)\end{equation}

\begin{equation}\label{allotherrplayp}
P_i(\text{All other $r$-players chose } P)=\frac{1}{(1-r_i)}\prod^{k}_{j=1}( 1-r_j)
\end{equation}
To simplify the next two items, let $W_{k'}$ be the subsets of variables $\{r_1,\dots r_{k}\}$ of size $k'$, and ${W^i}_{k'}$ the given sets without element $r_i$.

\begin{equation}\label{k'playerschooseR}
P(k' \text{ $r$-players chose } R)=\sum_{w\in W_{k'}}\prod_{r_j\in w} r_j\prod_{r_l\in w^c}(1-r_l)
\end{equation}
\begin{equation}\label{k'otherplayerschoser}
P_i( k' \text{ other $r$-players chose } R)=\sum_{w\in {W^i_{m'}}}\prod_{r_j\in w} r_j\prod_{r_l\in w^c}(1-r_l)
\end{equation}
For the final two items, we can use items $3$, and $4$ to simplify them as follows:
\begin{equation}
P( k'\text{ $r$-players chose } P)=P(( k-k') \text{ $r$-players chose } R)
\end{equation}
\begin{equation}
P_i(k' \text{ other $r$-players chose } P)= P_i( (k-1-k') \text{ other $r$-players chose } R)
\end{equation}

We can note that all non-$r$-player-specific probabilities above, i.e., those not of the form $P_i(X)$, are symmetric multi-linear polynomials in $r_1,\dots, r_m$, which implies that in these equations the coefficient of each of the products of $k'$ different $r_i$, $r_{j_1}\dots r_{j_{k'}}$, is the same regardless of the $r_i$s chosen. We can see that the same is true for the player-specific probabilities, $P_i(X)$, if we exclude $r_i$. Let $X^{b}$ be the sum of all products of $b$ distinct choices of $r_i$, $X_{j}^{b}$ to be these, not including the variable $r_j$, and similarly $X_{ij}^{b}$ these not including variables $r_i$ or $r_j$. Then:

\begin{equation}\label{k'rplayerschoser-2}
P(k' \text{ $r$-players chose } R)= \sum^{k}_{b=k'}(-1)^{b-k'}\binom{b}{k'}X^b
\end{equation}
We can see that this is the case by considering the coefficient of some element of the sum $X^{b'}$, $r_{j_1}\dots r_{j_{b'}}$ in equation \ref{k'playerschooseR}. An element of the sum in equation \ref{k'playerschooseR} only contributes to the coefficent of $r_{j_1}\dots r_{j_{b'}}$ if $w$ is contained in $j_1\dots j_{b'}$, and contributes $1$ or $-1$ depending only on whether $b'-k'$ is even or odd. Therefore the coefficient of $X^b$ must be  $(-1)^{b-k'}\binom{b}{k'}$ as shown above. Similarly, the following equations hold:

\begin{equation}\label{k'otherrplayerschoser-2}
P_i(k' \text{ other $r$-players chose } R)= \sum^{k-1}_{b=k'}(-1)^{b-k'}\binom{b}{k'}X_i^b
\end{equation}

\begin{equation}\label{k'rplayerschosep-2}
P(k'\text{ $r$-players chose } P)= \sum^{k}_{b=k-k'}(-1)^{b-(k-k')}\binom{b}{k-k'}X^b\end{equation}

\begin{equation}\label{k'otherrplayerschosep-2}
P_i(k'\text{ other $r$-players} \text{ chose } P)= \sum^{k-1}_{b=(k-1)-k'}(-1)^{b-((k-1)-k')}\binom{b}{(k-1)-k'} X_i^b
\end{equation}
Inputting these into the above-mentioned expected value equations and solving in terms of $X^b, X_i^b$ simplifies these expected values into:

\begin{equation} \label{eir-2}
E_i(R)=s\sum_{b=0}^{k-1}\frac {(-1)^b(1+k+t)}{1+b} X^b_i-(1-s)
\end{equation}
To see how we got this, plug in the equation \ref{k'otherrplayerschoser-2} into the equation \ref{Eir-1} to get:

\begin{equation}E_i(R)=s\sum_{k'=0}^{k-1} \frac{(k+t+1)-(k'+1)}{k'+1} \sum_{b=k'}^{k-1} (-1)^{b-k'} \binom{b}{k}X^b_i-(1-s)\end{equation}
Then we exchange the order of the sums to get:

\begin{equation}E_i(R)=s\sum_{b=0}^{k-1} (-1)^b X_i^b \sum_{k'=0}^b \frac{k+t+1-(k'+1)}{k'+1}(-1)^{-k'}\binom{b}{k}-(1-s)\end{equation}
The following identity then gives us equation \ref{eir-2}. 

\begin{equation}
\sum_{k'=0}^b \frac{k+t+1-(k'+1)}{k'+1}(-1)^{-k'}\binom{b}{k}=\frac{1+k+t}{1+b}
\end{equation}
We can similarly plug in equation \ref{k'playerschooseR} into \ref{Eip-1} to get:
\begin{equation}\label{Eip-2}
E_i(P)=(1-s) \sum_{b=0}^{k-1}\frac{1}{\binom{k+t}{b}}X^b_i -s
\end{equation}
This again comes from reordering the order of summation and using that:
\begin{equation}
    \sum_{k'=k-1-b}^{k-1} \frac{k-k'-1}{k+t+2}(-1)^{b-((k-1)-k')}\binom{b}{k-1-k'}=\frac{1}{\binom{k+t}{b}}
\end{equation}

Using the fact that equation \ref{allotherrplayp} is equivalent to $\sum_{b=0}^{k-1}(-1)^b X^b_i$ and plugging it into equation \ref{Eis-1} gets
\begin{equation}\label{eis-2}
E_i(S)=-1+(2-s)\frac{(k+1+t)}{2}\sum_{b=0}^{k-1}(-1)^b X^b_i
\end{equation}
If we take the same argument as above for $E_i(P)$ and exchange $k'$ in equation \ref{Eip-2} with $k''=k'+1$, then we have the result of plugging in equation \ref{k'rplayerschosep-2} into equation \ref{Esp-1}  and changing the order of summation. This results in:
\begin{equation}\label{esp-2}
E_s(P)=\sum_{b=0}^{k}\frac{1}{\binom{k+t}{b}}{X^b}\end{equation}
Similarly, we can apply the above identities to simplify the following:
\begin{equation}\label{ess-2}
E_s(S)=-1+(k+1+t)\sum_{b=0}^{k}(-1)^b X^b
\end{equation}

\begin{equation}\label{epr-2}
E_p(R)=s\sum_{b=0}^{k}\frac{(-1)^b(1+k+t)}{1+b} {X^b}-(1-s)
\end{equation}

\begin{equation}\label{epp-2}
E_p(P)=(1-s)\sum_{b=0}^{k}\frac{1}{\binom{k+t}{b}}{X^b}-s
\end{equation}

\begin{equation}\label{eps-2}
E_p(S)=-1+(2-s)\left(\frac{k+1+t}{2}\right)\sum_{b=0}^{k}(-1)^b {X^b}
\end{equation}
 
 Using the fact that $E_i(R)=E_i(P)$, and $E_j(R)=E_j(P)$, and subtracting the corresponding equations from each other results in:
\begin{equation}\label{big-eq-for-contra}
0=\sum_{b=0}^{k-1}\left(\frac{s(-1)^b(1+k+t)}{1+b} 
-\frac{1-s}{\binom{k+t}{b}}\right) (X^b_i-X^b_j)\end{equation}
In the difference $X^b_i-X^b_j$, the symmetric polynomial ${X^b}_{ij}$ cancels, leaving only terms $(r_j-r_i)X^{b-1}_{ij}$. Simplifying in this way results in:
\begin{equation}\label{big-eq-contra-2}
0=(r_j-r_i)\sum_{b=1}^{k-1}\left(\frac{s(-1)^b(1+k+t)}{1+b} 
-\frac{1-s}{\binom{k+t}{b}}\right) X^{b-1}_{ij}
\end{equation}
Therefore either $r_i=r_j$ or 
 \begin{equation}\label{big-eq-contra-3}
 0=\sum_{b=1}^{k-1}\left(\frac{s(-1)^b(1+k+t)}{1+b} 
-\frac{1-s}{\binom{k+t}{b}}\right) X^{b-1}_{ij}
\end{equation} 

As the second equation is multi-linear, it reaches all maxima and minima on the corner points, i.e., when variables $s, r_1\dots r_{m-1}$ are either $0$ or $1$.
As these equations are symmetric with respect to each $r_i$, we only need to consider the cases, $s=0$ or  $s=1$, and for $l\leq k-2$ values of $1$ and $k-2-l$ values of $0$ in $\{r_1\dots r_{k}\}\setminus \{r_i,r_j\}$. This results in:
\begin{equation}\label{big-eq-contra-4}0=\sum_{b=1}^{l+1}\left(\frac{s(-1)^b(1+k+t)}{1+b} 
-\frac{1-s}{\binom{k+t}{b}}\right) \binom{l}{b-1}
\end{equation} 
Which, when $s=0$ is simplified to:
\begin{equation}\label{big-eq-case1}
\frac{-(k+1+t)}{(-l+k+t)(-l+k+1+t)}
\end{equation}
and when $s=1$:
\begin{equation}\label{big-eq-contra-case-2}
\frac{-1-k-t}{(1+l)(2+l)}
\end{equation}
As  $l\leq k-2$, both of these are negative, and therefore equation \ref{big-eq-contra-3} is negative over the entire domain. Thus, $r_i=r_j$ for all $i,j$, proving our lemma above. We can then let variable $r=r_i$ for any $i$ and simplify:
\begin{equation}\label{x-simp}
    X^b=\binom{k}{b}r^{b}, \text{ and } X^b_i=\binom{k-1}{b}r^{b}
\end{equation}
This allows us to simplify the above equations of expected values further:
\begin{equation}\label{eir-3}
E_i(R)=s(1+k+t)\left( \frac{1-(1-r)^k}{kr}\right)-(1-s)\end{equation}

\begin{equation}\label{eip-3}
E_i(P)=(1-s) \left(\sum_{b=0}^{k-1}\frac{1}{\binom{k+t}{b}}\binom{k-1}{b}r^{b} \right) -s\end{equation}

\begin{equation}\label{eis-3}
E_i(S)=-1+(2-s)\frac{(k+1+t)}{2}(1-r)^{k-1}
\end{equation}

\begin{equation}\label{esp-3}
E_s(P)=\sum_{b=0}^{k}\frac{1}{\binom{k+t}{b}}\binom{k}{b}r^{b}\end{equation}

\begin{equation}\label{ess-3}
E_s(S)=-1+(k+1+t)(1-r)^k
\end{equation}

\begin{equation}\label{epr-3}
E_p(R)=s(1+k+t)\left( \frac{1-(1-r)^{k+1}}{(k+1)r}\right)-(1-s)
\end{equation}

\begin{equation}\label{epp-3}
E_p(P)=(1-s)\left(\sum_{b=0}^{k}\frac{1}{\binom{k+t}{b}}\binom{k}{b}r^{b}\right)-s
\end{equation}

\begin{equation}\label{eps-3}
E_p(S)=-1+(2-s)\left(\frac{k+1+t}{2}\right)(1-r)^k
\end{equation}
Of these equations, the expected values of playing paper are particularly complex, which makes it very hard to show conjecture \ref{numericallemmea} in general. Instead, we can solve this for fewer than $l$ players by checking all combinations of $k < l$, $ t < l$. Doing so reduces the above equations to polynomials, which a computer algebra system can solve. We do so in the Mathematica script \textit{RPS-50.wls} in the ancillary files of this submission, for $k\leq 50$ and $ t\leq 50$, resulting in no solutions.
\end{proof}

\section{The Imbalanced  $(m,2k+1)$-$RPS$ is Strongly Playable}

Using blow-ups and the imbalanced $(m,3)$-$RPS$ given above we can define an $(m,2k+1)$-$RPS$ as  $(m,3)$-$RPS\#_{S}((m,3)$-$RPS)\#_{S}\dots (m,3)$-$RPS)\dots)$. Just as we showed for the $(m,3)$-$RPS$, this game reaches the maximal imbalance for any imbalance in the symmetric Schur-distributional and Schur-uniform imbalance classes in the limit as $m$ goes to infinity. 
For this reason, we call this construction the imbalanced $(m,2k+1)$-$RPS$. We provide the following conjecture that will be necessary to assume in proving that this game is strongly playable:

\begin{conj} \label{sadconjecture}
    In any Nash equilibrium of the given $(m,2k+1)$-$RPS$ game, $N$, if the probability of player $i$ playing $S$ is $P_i(S)$ and the probability of $i$ playing any $P$-type object is $P_{i}(T_p)$, then we conjecture that:
    \begin{equation}
        \frac{P_i(T_p)}{P_i(S)}\geq (m-1)
    \end{equation}
    Moreover, for any $k>1$, and player $q_i$, the expected number of other players tying with $q_i$ for the win that occur for $P_k$ is less than $m-2$.
\end{conj}

The sketch of the reason for this conjecture is that as the number of players increases, the probability that a player plays $S$ in the symmetric Nash equilibrium goes down rapidly in the $3$ object game. To the point that with $20$ players, this ratio is over $50$. As the number of objects increases, we would expect this to occur less often. Moreover, we would expect that objects in higher-level blow-ups have a smaller proportion of play and that, as the number of players increases, the expected number of ties for a win when $q_i$ plays one of these objects is much less than all but one other player.

Assuming this conjecture, we can  prove the following: 

\begin{thm}
    The imbalanced $(m,2k+1)$-$RPS$ is a playable $RPS$ on $(2k+1)$ objects. Moreover, if we continued this construction infinitely, we would have a $(m,\N)$-$RPS$ that is also a playable $RPS$. Moreover, these games are strongly playable for up to $50$ players. 
\end{thm}

\begin{proof}
Before we begin, we will set some terms so as to clarify the proof. Objects in the $i$th blow up that are not in the $(i+1)$th blow up will be referred to as objects in level $i$. For example, the level $1$ objects are $R$ and $P$ in the game, and the level $k$ objects are $R_k,P_k$, and $S$. We will also refer to the imbalanced $(m,2i+1)$-$RPS$ as $RPS_i$. We will first demonstrate strong probability for fewer than $50$ players, and in doing so, provide an argument for the existence of a symmetric Nash equilibrium to prove playability for any number of players.

Above, we showed that $RPS_3$ is strongly playable and that at least two players must play $S$ in this game with fewer than $50$ players. Assume by induction that for some $j$, $RPS_j$ is strongly playable where at least two players play $S$. If we can show that $RPS_{j+1}$ is strongly playable where at least two players play $S$, the theorem is proven for all $k\in \N$. Note that for the symmetric case, it will only be necessary to assume that at least one person plays $S$ in each Nash equilibrium of $RPS_j$. The base case for this is true for any number of players.   

We will start by showing the following lemma:

\begin{lem}
    For any $k\in \N, m<50$, in any Nash equilibrium, $N$, of $RPS_k$, at least two players play some objects in level $k$ with positive probability. For any symmetric Nash equilibrium of $RPS_k$, $N_s$, all players play some object in level $k$ with positive probability. 
\end{lem}

\begin{proof}

Firstly, we must show that in $RPS_j$ at least two players must play an object at level $(j+1)$. By contradiction, assume that in a Nash equilibrium, $N$, at most one player plays in level $(j+1)$. Then, we can project each probability vector of each player to probability vectors in $RPS_j$. Let $P_j(N)$ be the projection. These mixed strategies cannot be a Nash equilibrium as, at most, a single player plays $S$. Therefore, in $RPS_j$, some player $q$ can improve their expected payoff by changing their probability vector $p(v_q)$ to some new probability vector $v'_q$. Define $p^{-1}(v'_q)$ as the probability vector in $RPS_{j+1}$ that agrees with $v'_q$ on all probabilities for objects in $RPS_j$, and if no other player played in level $(j+1)$, then $p^{-1}(v'_q(S))$ gives all remaining probability to the object $S$ in said $(m,3)$-$RPS$. If another player played in this level, then $p^{-1}(v'_q(S))$ exactly mirrors their relative proportions for $r_k,p_k,S$ such that:
\begin{equation}
p^{-1}(v'_q(S))(R_k)+p^{-1}(v'_q(S))(P_k)+p^{-1}(v'_q(S))(S)=v'_q(S)
\end{equation}

When exactly two people play at this level, this strategy mirrors the outcome on average as if both players played $S$ in level $k$, with their proportions corresponding to the sum of their proportions for $S, R_{k+1}, P_{k+1}$. Therefore, $q$'s expected payoff with $v_q$ is the same as their expected payoff with $p(v_q)$, which is strictly less than their expected payoff with $v'_q$, which is the same as their expected payoff with $p^{-1}(v'_q(S))$. So, $N$ cannot be a Nash equilibrium.

For a symmetric Nash equilibrium, $N_s$, either no one plays at this level or everyone does. However, by a similar induction, this produces a Nash equilibrium of $RPS_{j}$ where no one plays $S$, which is a contradiction. In such a case, we can conclude that all players play at level $(j+1)$. 
\end{proof}

We next prove the following lemma, which, under the assumptions of the previous lemma, shows that all objects in levels $l<k+1$ must be played by some player with positive probability in every Nash equilibrium. 

\begin{lem}[Playing at level $k$ implies playing at all other levels]
    In a Nash equilibrium, $N$, if at least two players, $q_0\dots q_i$, play at level $(k+1)$ with positive probability, then every object in levels $l<k+1$ must be played with positive probability.
\end{lem}

\begin{proof}
By contradiction, in Nash equilibrium, $N$, assume $l<k+1$ is the highest level with an object with no probability of being played. As the $(k+1)$ level has positive probability of being played by at least two people, we have three cases: 
\begin{enumerate}
    \item In level $l$, all objects have a zero probability of being played.
    \item The probability of $P_l$ being played is positive and the probability of $R_l$ being played is $0$.  
    \item The probability of $R_l$ being played is positive and the probability of $P_l$ being played is $0$.
\end{enumerate}

For case $1$, if level $l$ has nothing being played, then by assumption, there are at least two people playing in higher levels. For player $q_0$ playing at a higher level than $l$, moving their probability for all higher level items to $R_l$ is weakly dominating, as in any situation in which $q_0$ would have won before, they still win and get at least as high a payoff, and in any situation in which $q_0$ loses, they get the same payoff. As other players $q_1\dots q_i$ also play objects with positive probability above level $l$, they must play objects that win or tie some portion of the time when $q_0$ also plays at this level. If not, then by symmetry, any one of these players, $q_j$, could mirror $q_0$'s strategy. If they were to do so, then in a scenario when an object at or above this level would have won previously, and only $q_0$ is also playing an object at or above this level, then player $q_j$ would get $\frac{m-1}{2}$ expected value instead of $-1$. In a scenario where an object at or above this level would win, and only they and a subset of other players in $q_1\dots q_i$  are playing at or above this level, then player $q_j$ would receive an expected payoff of $m-1$. Finally, previously in any scenario in which an object at or below this level would have won, and $q_0$ and some non-empty subset of other players $q_1\dots q_i$ play player $q_j$'s expected value before changing strategies was $-1$. Therefore, their expected value in any outcome in which they play at or above this level increases after this change, with a strict increase in certain scenarios with non-zero probability. Thus, case $1$ could not occur in a Nash equilibrium.

For case $2$, if $P_l$ has a probability of being played but $R_l$ does not, then any player playing $P_l$, $q_0$, moving their probability to an object in higher levels, is weakly dominating, as they would lose in any scenario when at least one other player plays any other object above this level. In the scenario when all other players play only $P_{l+1}$ at or above this level, then by moving the probability from $P_l$ to $S$, the expected outcome for $q_0$ when an object at or above this level wins and $q_0$ and another player plays at or above this level goes from $-1$ to $m-1$. In any other scenario, when the other players play some set of strategies at or above this level, excluding $R_l$, then $q_0$ playing $R_{l+1}$ is a strict improvement. To see this, consider any scenario where at least one player other than $q_0$ also played at this level, and they played $P_{l+1}$. In the previous case, the expected outcome is $-1$, and in this case, it is at least $-1$. However, in any scenario where $q_0$ playing $P_l$ would have won, player $q_0$ playing $R_{l+1}$ now wins alone for a higher than or equal expected value. Moreover, in the scenario where $q_0$ plays $R_{l+1}$ instead of $P_l$ and any object other than $P_l$, and $P_{l+1}$ at or above level $l$ would have won, now $R_{l+1}$ wins. As in this case $q_0$ would have lost before, they have a strictly higher expected value in this scenario. Therefore, $q_0$'s expected payoff in any outcome in which they play at or above this level increases after this change, with a strict increase in certain scenarios with non-zero probability. Thus, case $2$ could not occur in a Nash equilibrium.

For case $3$, if $R_l$ had a probability of being played but $P_l$ did not, then playing any object at levels above $l$ is weakly dominated by playing $R_l$. As at least two players play at the $(k+1)$ level, if player $q_0$ playing at this level shifted his probability of playing above level $l$ to $R_l$, this would be a strict improvement. To see this consider any instance in which $q_0$ and at least one other player would have played above level $l$, and such an object won. Then after changing strategies, $q_0$ would win. As before, if $q_0$ had already won in each of these scenarios, then any other player mirroring $q_0$'s strategy would receive a net higher expected value, and that could also not be a Nash equilibrium. Thus, case $3$ could not occur in any Nash equilibrium.
\end{proof}

Therefore, until the $(k+1)$ level, every object must have some positive probability of being played. In level $(k+1)$, as at least two people play in this $(m,3)$-$RPS$, if $S$ had no probability of being played, then all players would play $P_{k+1}$ at this level, but then playing $S$ would have a higher expected value as whenever more than two players play in level $(k+1)$, the player playing $S$ would always win. Similarly, if $R_{k+1}$ is not played, then $S$ would be the only object at level $(k+1)$ that is played, and if $P_{k+1}$ is not played, then $R_{k+1}$ would be the only object that is played again leading to a contradiction as in the base $(m,3)$-$RPS$ case. 

As this holds for the symmetric Nash equilibrium with any number of players, this finishes the argument that $(m,2k+1)$-$RPS$ is playable by showing that the symmetric equilibria are all playable. We have also shown under the inductive assumption that $RPS_{k+1}$ is strongly playable with under $50$ players, leaving only to show that at least two players play $S$ in this game. 
Consider by contradiction that at most one player, $q_0$, played $S$ in the $(k+1)$-level. We will show that no player would then play $R_{k+1}$, which contradicts the argument above. We can see that player $q_0$ never plays $R_{k+1}$ because it would only win if an object at this level won, and all other players at this level only played $R_{k+1}$. In any such scenario, $q_0$ playing $P_{k+1}$ has a strictly higher expected value.

By the conjecture, \ref{sadconjecture}, the ratio of the proportion said player plays $S$, $P_0(S)$, to the proportion the play an $P$-type object, $P_0(T_p)$ is:
\begin{equation}
    \frac{P_0(S)}{P_0(T_p)}\geq \frac{1}{m-1}
\end{equation}

Given that an object in level $(k+1)$ won, it must be the case that player $q_0$ played either a $P$-type object or $S$. For any other given player $q_i$, let the proportion of the time, that they would have won if they played $P_{k+1}$ given that, if they played in this level an object in this level would have won, be given by $P(W_{P_{k+1}}|W_{k+1})$, and similarly define $P(W_{R_{k+1}}|W_{k+1})$. Moreover, we can see that if $q_i$ chooses an object in this level, it would win only if all other players play in this level or $P$-type objects. Furthermore, given that all players other than $q_0$, and $q_i$ play $P$-type objects or objects in this level, as each player plays independently, the scenarios in which $R_{k+1}$ would have won are precisely those in which $q_0$ plays $S$ and $P_{k+1}$ would have won in exactly the situations in which $q_0$ plays any other $P$-type object. We can then see that:
\begin{equation}
    \frac{P(W_{P_{k+1}}|W_{k+1})}{P(W_{R_{k+1}}|W_{k+1})}=\frac{P_0(S)}{P_0(T_p)}\geq \frac{1}{m-1}
\end{equation}

Let $V_{R_{k+1}}$,$V_{P_{k+1}}$ be the expected payoff for $R_{k+1}$ if it wins and $P_{k+1}$ if it wins respectively. We can then see that, given that all players other than $q_i$ play $P$-type objects or objects in this level, the expected payoff for $q_i$ playing $R_{k+1}$, $P_{k+1}$ is:

\begin{equation}
    \begin{split}
    E_i(R_{k+1}|W_{k+1})= V_{R_{k+1}}\left(\frac{1}{m}\right)-1\left(\frac{m-1}{m}\right) \\
    E_i(P_{k+1}|W_{k+1})= V_{P_{k+1}}\left(\frac{m-1}{m}\right)-1\left(\frac{1}{m}\right)    
    \end{split}
\end{equation}
As by the conjecture \ref{sadconjecture}, $\frac{1}{m-1}<V_{P_k}$, and by construction $V_{P_{k+1}}\leq m-1$, we have that:
\begin{equation}
    \begin{split}
         E_i(P_{k+1}|W_{k+1}) &>\frac{1}{m-1}\left(\frac{m-1}{m}\right)-1\left(\frac{1}{m}\right)\\
         &=0\\
         &=(m-1)\left(\frac{1}{m}\right)-1\left(\frac{m-1}{m}\right)\\
         &\geq E_i(R_{k+1}|W_{k+1})
    \end{split} 
\end{equation}
We can then see that for all players other than $q_0$, $E_i(P_{k+1}|W_{k+1})>E_i(R_{k+1}|W_{k+1})$, and therefore $E_i(P_{k+1})>E_i(R_{k+1})$. Thus, no player plays $R_{k+1}$ in any such Nash equilibrium, and we have a contradiction.

\end{proof}

Finally, we will show that $(\N,m)$-$RPS$ is playable for all $m\in \N$, and strongly playable for $m<50$. 

\begin{proof}
Consider the case that in some Nash equilibrium, $N$, there is a level $l$ such that all higher levels are empty. Then the strategies in the $(\N,m)$-$RPS$ map over exactly to strategies and expected payoffs in the imbalanced $(2l+3,m)$-$RPS$, where no one plays in level $(l+1)$. If there were fewer than $50$ players or if the equilibrium was symmetric, then this is a contradiction. If, on the other hand, there is no such level, then by contradiction, assume that in $N$, $l$ is the highest level with an object with no probability of being played.

As before, we have the three cases: 
\begin{enumerate}
    \item In this level, all objects have a zero probability of being played.
    \item The probability of $P_l$ being played is positive and the probability of $R_l$ being played is $0$.
    \item The probability of $R_l$ being played is positive and the probability of $P_l$ being played is $0$.
\end{enumerate}

Each of these cases resolves as they have in the finite case, with the only change being that instead of choosing $S$, one could choose $R_f$ for some high enough $f$ whenever applicable. Therefore, we can conclude that $(\N,m)$-$RPS$ is playable and that given conjecture \ref{sadconjecture} with less than $50$ players is strongly playable.
\end{proof}

\phantomsection
\bibliographystyle{amsplain}
\bibliography{references}

\end{document}